\documentclass[pdftex,twocolumn,epjc3]{svjour3}          

\usepackage[]{changes} 

\RequirePackage[T1]{fontenc}

\smartqed  
\usepackage{xurl}
\RequirePackage{graphicx}
\RequirePackage{mathptmx}      
\RequirePackage{flushend}
\RequirePackage[numbers,sort&compress]{natbib}
\RequirePackage[colorlinks,citecolor=blue,urlcolor=blue,linkcolor=blue]{hyperref}
\usepackage[utf8]{inputenc}
\usepackage{amssymb,amsmath,amsfonts}
\usepackage{enumerate,mathrsfs}
\usepackage{xcolor}
\usepackage{cancel}
\usepackage[squaren]{SIunits}
\usepackage{mathtools}
\usepackage{prettyref}
\usepackage{blkarray}
\usepackage{adjustbox}
\usepackage{array,booktabs,tabularx,fancybox}
\usepackage{comment}
\usepackage[normalem]{ulem}
\usepackage{afterpage}

\newtheorem{prop}{Proposition}

\newtheorem{cor}{Corollary}

\begin{document}

\title{Geometry of the neutrino mixing space}


\author{Wojciech Flieger\thanksref{e1,addr1}
        \and
        Janusz Gluza\thanksref{e2,addr2} 
}

\thankstext{e1}{e-mail: flieger@mpp.mpg.de}
\thankstext{e2}{e-mail: janusz.gluza@us.edu.pl}

\institute{Max-Planck-Institut f\"ur Physik, Werner-Heisenberg-Institut, 80805 M\"unchen, Germany\label{addr1}
          \and
          Institute of Physics, University of Silesia, Katowice, Poland\label{addr2}
}

\date{}

\maketitle

\begin{abstract}

We study a geometric structure of a physical region of neutrino mixing matrices as part of the unit ball of the spectral norm.
Each matrix from the geometric region 
is a convex combination of unitary \texttt{PMNS} matrices. The  disjoint subsets corresponding to a different minimal number of additional neutrinos are described as relative interiors of faces of the unit ball. We determined the Cara\-th\'eo\-dory's number showing that, at most, four unitary matrices of dimension three are necessary to represent any matrix   from the neutrino geometric region. For matrices  which correspond to scenarios with one and two additional neutrino states, the Carath\'eodory's number is two and three, respectively. 
Further, we discuss the volume associated with different mathematical structures, particularly with unitary and orthogonal groups, and the unit ball of the spectral norm.
We compare the obtained volumes to the volume of the region of physically admissible mixing matrices for both the CP-conserving and CP-violating cases in the present scenario with three neutrino families and scenarios with the neutrino mixing matrix of dimension higher than three.
\end{abstract}

\section{Introduction}
Over the years neutrino oscillation experiments have provided in-depth information about the structure of the neutrino standard $3\times 3$ unitary mixing matrix $U_{\texttt{PMNS}}$ \cite{Zyla:2020zbs, Esteban:2020cvm}. We know already that the $\theta_{13}$ mixing angle is nonzero, and as a consequence, the (1,3) element of $U_{\texttt{PMNS}}$ is also nonzero \cite{DayaBay:2012fng,RENO:2012mkc,DoubleChooz:2012gmf}. In that way, the tri-bimaximal mixing structure has been excluded \cite{Harrison:2002er}. Recently a lot of attention  is  given to the study of the value of the neutrino CP complex phase. If it is nonzero it could shed a new light on the matter-antimatter problem \cite{T2K:2019bcf}. 
On top of that, there is a possibility that more than three known neutrinos exist. In this case new neutrino states, commonly known as sterile neutrinos, can mix with active Standard Model neutrinos. This implies that the $3 \times 3$ neutrino mixing matrix is no longer unitary. There are various approaches to deal with the non-unitarity problem, for instance 
{a decomposition of a general matrix into a product with a unitary matrix are considered. The two often used approaches are known as the $\alpha$ and $\eta$ parametrizations \cite{Antusch:2006vwa, FernandezMartinez:2007ms, Xing:2007zj, Xing:2011ur, Escrihuela:2015wra, Blennow:2016jkn}. In the $\alpha$ parametrization's framework a small deviation from unitarity is encoded into a lower triangular matrix,  whereas in the $\eta$ framework, possible deviations from unitarity are encoded in a Hermitian matrix.} 
In \cite{Bielas:2017lok} a different approach was proposed based on a matrix
theory where the interval neutrino mixing matrix $U_{int}$ is studied using matrix theory methods, and its connections to the non-standard neutrino
physics have been established by exploring singular values and contractions. 
In \cite{Flieger:2019eor} the matrix theory has been applied to phenomenological studies and new limits on light-heavy neutrino mixings in the $3+1$ model (three light, known neutrinos with one additional sterile neutrino) have been obtained. 
In another work where the matrix theory methods have been explored in the context of neutrino physics, 
conditions for the existence of the gap in the seesaw mass spectrum  {have} been established and {justified}  \cite{Besnard:2016tcs,Flieger:2020lbg}. 
We should also mention that an interesting mathematical connection between eigenvalues and eigenvectors has been rediscovered in the context of neutrino oscillations in matter \cite{Denton:2019ovn,Denton:2019pka}.
There has been also attempts to predict neutrino masses by geometric and topological methods. In \cite{Asselmeyer-Maluga:2018ywa} the neutrino mass spectrum is explored through the model of cosmological evolution based on the exotic smooth structures.  
In this work we focus on further geometrically based studies towards the understanding of the class of physically admissible neutrino mixing matrices where the $3\times 3$ mixing matrix could be a part of a higher dimensional unitary mixing matrix. 
Our aim is to study the structure of a geometric region $\Omega$ that corresponds to physically admissible mixing matrices, which has been introduced in \cite{Bielas:2017lok}. 

In the next chapter we give a general setting for our discussion introducing a neutrino mixing matrix its most common parametrization and current experimental limits for the mixing parameters.
In the third chapter, we define the region of physically admissible mixing matrices and its subsets corresponding to a different minimal number of additional neutrinos. 
In the fourth chapter, we recognize the geometric region as a subset of the unit ball of a spectral norm and describe its facial structure. 
Next, we will connect the facial structure of $\Omega$ to the minimal number of additional sterile neutrinos and we will determine the so-called Carath\'eodory number which informs us about a minimal number of unitary $3 \times 3$ matrices which are needed to span the whole $\Omega$ space for a given number of sterile neutrinos. This allows for an optimal construction of physical mixing matrices that can be used for further analysis of scenarios involving sterile neutrinos.
Finally, we determine the volume of this region for CP conserving and violating cases. 
The article is finished with a summary and outlook. 
The main text is supported by the Appendix containing auxiliary definitions and theorems.

\section{Setting: Neutrino mixing matrix and experimental data}
Neutrino flavour fields are (linear) combinations of the massive fields
\begin{equation}
\nu^{f}_{l} = \sum_{i=1}^n U_{li} \nu^{m}_{i}   .
\label{eq:mix}
\end{equation}
This property of neutrino fields is called the neutrino mixing mechanism. The mixing of neutrinos occurs regardless if they are Dirac or Majorana particles \cite{Bilenky:1980cx,Gluza:2016qqv}.  
As the massive and flavour fields form two orthogonal bases in the state space, the transition from one base to another can be done by the unitary matrix. This restricts coefficients of the linear combination, the sum of squares of their absolute values must equal one 
\begin{equation}
\nu^{f}_{l} = \sum_{i=1}^n U_{li} \nu^{m}_{i} \quad \text{ with } \quad \sum_{i=1}^n \vert U_{li} \vert^{2}=1.
\end{equation}
For $n=3$ in (\ref{eq:mix}), the $3\times 3$ unitary matrix $U$ corresponding to three light known neutrino mixing is known as the \texttt{PMNS} mixing matrix \cite{Pontecorvo:1957qd,Maki:1962mu}. 
The general $n \times n$ complex matrix has $n^2$ complex parameters or equivalently $2n^2$ real parameters. The unitarity condition $UU^{\dag} =I$ imposes additional $n^{2}$ constraints on the elements. It can be seen from the $UU^{\dag}$ which is a Hermitian matrix and has $n$ independent diagonal elements and $n^{2}-n$ independent off-diagonal elements which together give $n^{2}$ independent elements or conditions imposed on the unitary matrix. Thus, the $n \times n$ unitary matrix has $2n^{2} - n^{2} = n^{2}$ independent real parameters. An alternative way to see this is by writing a unitary matrix as the matrix exponent of the Hermitian matrix, i.e. $U = e^{iH}$, where the $H$ matrix is Hermitian and thus has $n^{2}$ independent real parameters which implies that $U$ also has $n^{2}$ independent real parameters. These parameters can be split into two categories: rotation angles and complex phases. The number of angles corresponds to the number of parameters of the orthogonal matrix which has $\frac{n(n-1)}{2}$ independent real parameters. The remaining parameters correspond to phases. Thus, the $n^{2}$ independent real parameters of the unitary matrix split into
\begin{equation}
\begin{split}
&\text{angles: } \quad \frac{n(n-1)}{2}, \\
&\text{phases: } \quad  \frac{n(n+1)}{2}.
\end{split}    
\end{equation}
However, not all phases are physical observables. The charged leptons and neutrino fields can be redefined as
\begin{equation}
\nu_{i} \rightarrow e^{i \alpha_{i}} \nu_{i} \text{ and } l \rightarrow e^{i \beta_{l}} l.    
\end{equation}
The $\alpha_{i}$ and $\beta_{l}$ phases can be chosen in such a way that they eliminate $2n-1$ phases from the mixing matrix leaving the Lagrangian invariant. This reduces the number of phases of the mixing matrix. The number of remaining free parameters is $(n-1)^{2}$ which divides into
\begin{equation}
\begin{split}
&\text{angles: } \quad \frac{n(n-1)}{2}, \\
&\text{phases: } \quad  \frac{(n-1)(n-2)}{2}.
\end{split}    
\label{dirac_para}
\end{equation}
These are the numbers under consideration when neutrinos are of the Dirac type. However, we know already that neutrinos can also be particles of the Majorana type. Then the Majorana condition $\nu_{i}^{\mathcal{C}}=\nu_{i}$ where ${\mathcal{C}}$ is the charge conjugate operator, 
fixes phases of the neutrino fields, which no longer can be chosen to eliminate phases in the mixing matrix. On the other hand, the phases of charged leptons are still arbitrary and can be chosen in such a way as to eliminate phases from the mixing matrix. Thus, from all $\frac{n(n+1)}{2}$ phases of the unitary matrix, $n$ phases can be eliminated. Finally, for the Majorana neutrinos, the number of free parameters of the mixing matrix is as follows 
\begin{equation}
\begin{split}
&\text{angles: } \quad \frac{n(n-1)}{2}, \\
&\text{phases: } \quad  \frac{n(n-1)}{2}.
\end{split}   
\label{majorana_para}
\end{equation}

Knowing the number of parameters necessary to describe the mixing matrix, we can find its explicit form by invoking a particular parametrization. In the minimal scenario, the mixing matrix is a $3 \times 3$ matrix and thus for the Dirac case we have three mixing angles and one complex phase.
The standard way of parametrizing the \texttt{PMNS} mixing matrix is as the product of three rotation matrices with additional complex phase in one of them, i.e. in terms of Euler angles $\theta_{12}$, $\theta_{13}$, $\theta_{23}$ and complex phase $\delta$

\begin{small}
\begin{equation}
\label{pmns_para}
\begin{split}
U_{\texttt{PMNS}}&=
\left(
\begin{array}{ccc}
1 & 0 & 0 \\
0 & c_{23}  & s_{23} \\
0 & -s_{23} & c_{23}
\end{array}
\right)
\left(
\begin{array}{ccc}
c_{13} & 0 & s_{13} e^{-i\delta} \\
0 & 1 & 0 \\
-s_{13} e^{i\delta} & 0 & c_{13}
\end{array}
\right)
\left(
\begin{array}{ccc}
c_{12} & s_{12} & 0 \\
-s_{12} & c_{12} & 0 \\
0 & 0 & 1
\end{array}
\right) \\
&\equiv
\left(
\begin{array}{ccc}
U_{e1} & U_{e2} & U_{e3} \\
U_{\mu 1} & U_{\mu 2} & U_{\mu 3} \\
U_{\tau 1} & U_{\tau 2} & U_{\tau 3}
\end{array}
\right).
\end{split}
\end{equation}
\end{small}

In the case of Majorana neutrinos we must include additional phases, which is done typically by multiplying the \texttt{PMNS} mixing matrix from the right-hand side by the diagonal matrix of phases $P^{M}$. For the $3 \times 3$ mixing matrix, we must add two more complex phases. The Majorana neutrino mixing matrix is then given by
\begin{equation}
U_{\texttt{PMNS}}^{M} = U_{\texttt{PMNS}}P^{M}, \text{ where } P^{M}=diag(e^{i\gamma_{1}}, e^{i\gamma_{2}},1).    
\end{equation} 

The oscillation experiments provide the major information about the structure of the neutrino mixing matrix.
The current data gives the following limits for the mixing parameters \cite{Zyla:2020zbs, Esteban:2020cvm}
\begin{equation}
\label{pmns_limits}
\begin{split}
&\theta_{12} \in [31.27^{\circ}, 35.86^{\circ}] , \quad \theta_{23} \in [40.1^{\circ}, 51.7^{\circ}], \\
&\theta_{13} \in [8.20^{\circ}, 8.93^{\circ}] , \quad \quad \delta \in [120^{\circ}, 369^{\circ}].
\end{split}    
\end{equation}

By inputting these ranges into \eqref{pmns_para} we get allowed ranges for the mixing matrix elements \cite{Esteban:2020cvm} (at the $3 \sigma$ confidence level)

\begin{equation}
\vert U \vert_{3\sigma} =
\left(
\begin{array}{ccc}
\left[ 0.801,0.845 \right] & \left[ 0.513,0.579 \right] & \left[ 0.143,0.155 \right] \\
\left[ 0.243,0.500 \right] & \left[ 0.471,0.689 \right] & \left[ 0.637,0.776 \right] \\
\left[ 0.271,0.525 \right] & \left[ 0.477,0.694 \right] & \left[ 0.613,0.756 \right]
\end{array}
\right).    
\end{equation}

The exact values of the allowed ranges in the CP invariant case presented as the interval matrix are 
\begin{equation} 
\begin{split}
&O_{int} = \\ 
&\left(
\begin{array}{ccc}
\left[0.801,0.845\right] & \left[0.513,0.579\right] & \left[0.143,0.155\right] \\
\left[-0.529,-0.417\right] & \left[0.431,0.606\right] & \left[0.637,0.776\right] \\
\left[0.233,0.388\right] & \left[-0.721,-0.586\right] & \left[0.613,0.756\right]
\end{array}
\right),
\label{u_int_real}
\end{split}
\end{equation}
whereas when the non-zero CP phase $\delta$ is included, the elements of the $U_{int}$ are within the following ranges
\begin{equation}
\begin{split}
&U_{e1} \in \left[ 0.801, 0.845 \right], \\
&U_{e2} \in \left [0.513, 0.579 \right], \\
&U_{e3} \in \left[ -0.155 -0.155 i, 0.155 + 0.134 i \right], \\
&U_{\mu 1} \in \left[ -0.528 - 0.0901 i, -0.218 + 0.104 i \right], \\
&U_{\mu 2} \in \left[ 0.432 - 0.0616 i, 0.707 + 0.0711 i \right], \\
&U_{\mu 3} \in \left[ 0.637, 0.776 \right], \\
&U_{\tau 1} \in \left[ 0.233 - 0.0878 i, 0.538 + 0.101 i \right], \\
&U_{\tau 2} \in \left[ -0.721 - 0.060 i, -0.453 + 0.0693 i \right], \\
&U_{\tau 3} \in \left[ 0.613, 0.756 \right].
\end{split}  
\label{u_int_complex}
\end{equation}

Though the experimental results given in (\ref{u_int_real}) and (\ref{u_int_complex}) are based on the $U_{\texttt{PMNS}}$ matrix (\ref{pmns_para}), we can reverse the problem and ask the following question: What can we learn about a geometrical structure of a region of physical mixing matrices, not restricted to $U_{\texttt{PMNS}}$, given basic mixings between light known neutrinos in (\ref{u_int_real}) or (\ref{u_int_complex})? 

We will answer this question in the following sections. 

\section{Region of physically admissible mixing matrices}\label{chap:phys_reg}
We are interested in a special class of matrices encompassing unitary matrices or matrices which can be a submatrix of a unitary matrix. These are known as contractions and are defined by the following formula $\Vert A \Vert \leq 1$ (for necessary definitions see \ref{appnorms}). 
The importance of contractions in neutrino mixing studies and their properties have been discussed in \cite{Bielas:2017lok} and \cite{Flieger:2019eor}.

We will show that a matrix constructed as a finite convex combination of unitary matrices is a contraction. 
Let $U_{i}$, $i=1,\dots,n$, be a unitary matrix, and let $A = \sum_{i=1}^{n}\alpha_{i}U_{i}$ with \\
$\alpha_{i} \geq 0$ and $\sum_{i=1}^{n}\alpha_{i}=1$, then
\begin{equation}
\Vert A \Vert = \Vert \sum_{i=1}^{n}\alpha_{i}U_{i} \Vert \leq \sum_{i=1}^{n}\alpha_{i} \Vert U_{i} \Vert = \sum_{i=1}^{n}\alpha_{i} = 1 \Rightarrow \Vert A \Vert \leq 1.  
\end{equation}
The converse is also true \cite{zhang11}, thus we have 
\begin{theorem}
A matrix $A$ is a contraction if and only if $A$ is a finite convex combination of unitary matrices.
\end{theorem}

This characterization of contractions has physical consequences. It allows gathering physically meaningful mixing matrices into a geometric region.

\begin{definition}
The region of all physically admissible mixing matrices, denoted $\Omega$, is the set of all finite convex combinations of $3 
\times 3$ unitary matrices with parameters restricted by experiments
\begin{equation}
\begin{split}
\Omega :=& conv(U_{\texttt{PMNS}}) = \\
=& \lbrace \sum_{i=1}^{m} \alpha_{i}U_{i} \mid U_{i} \in U(3), 
\alpha_{1},...,\alpha_{m} \geq 0, \sum_{i=1}^{m}\alpha_{i}=1, \\
&\theta_{12}, \theta_{13}, \theta_{23} \ \begin{rm}and\end{rm} \ \delta \ \begin{rm}given \ by\end{rm} \ \rm{experimental \ values} \rbrace.
\end{split}
\label{eq:omegaconv}
\end{equation}
\end{definition}
There is another equivalent definition of the $\Omega$ region, which reflects its geometric nature, namely as the convex hull spanned on the unitary \texttt{PMNS} matrices. 

The Corollary \ref{min_dim} given in \ref{appCS}
restricts the minimal dimension of the unitary extension of the contractions. This allows us to divide the $\Omega$ region into four disjoint subsets according to the minimal dimension of the unitary dilation
\begin{eqnarray}
     \Omega_{1}:&& \text{3+1 scenario: } \Sigma = \lbrace \sigma_{1}=1.0, \sigma_{2}=1.0,   \sigma_3 < 1.0 \rbrace  , \nonumber \\  && \label{omega_1}  \\
     \Omega_{2}:&& \text{3+2 scenario: } \Sigma = \lbrace \sigma_{1}=1.0, \sigma_{2}<1.0, \sigma_3 < 1.0 \rbrace \ , \nonumber \\  && \label{omega_2}  \\
    \Omega_{3}:&& \text{3+3 scenario: } \Sigma = \lbrace  \sigma_{1}<1.0, \sigma_{2}<1.0, \sigma_3 < 1.0 \rbrace  ,  \nonumber \\  && \label{omega_3}  \\
    \Omega_{4}:&& \texttt{PMNS} \text{\; scenario: } \Sigma = \lbrace \sigma_{1}=1, \sigma_{2}=1, \sigma_{3}=1 \rbrace \label{omega_4} .
\end{eqnarray}

This division allows to analyze individually scenarios with a different number of sterile neutrinos.
Thus, the study of geometric features of this region gives a possibility for a better understanding of neutrino physics, especially regarding the number of additional sterile neutrinos and the structure of the complete mixing matrix. 

It is important to notice that matrices from the $\Omega_{1}$ subset can be extended to unitary matrices of arbitrary dimension, starting from the dimension four. The same is true for contractions from the subset $\Omega_{2}$ which can produce any unitary matrices of dimension five or higher. It may look that there is an overlapping between matrices from different subsets of the $\Omega$ region and some of them may be redundant. {\it This is however not true, as unitary matrices produced by the contraction from each subset are unique.} It is so because contractions must end up in the $3 \times 3$ top diagonal block of a complete unitary matrix and as the subsets are disjoint, we cannot reproduce the same unitary matrices using contractions from different subsets. Thus, instead of overlapping, we should treat dilations of a given dimension of contractions from different subsets as complementary to each other.

\section{Geometry of the region of physically admissible mixing matrices}
The $\Omega$ region is a subset of the unit ball of the spectral norm
\begin{equation}
\label{unit_ball}
\mathcal{B}(n) = \lbrace A \in \mathbb{C}^{n \times n}: \Vert A \Vert \leq 1\rbrace.
\end{equation}

This fact allows us to give another characterization of the $\Omega$ region as the intersection of the $\mathcal{B}(3)$ with the interval matrix $U_{int}$ in Eqs.~(\ref{u_int_real})-(\ref{u_int_complex}), i.e.
\begin{equation}
\Omega = \mathcal{B}(3) \cap U_{int}.
\label{omega_intersection}
\end{equation}
{The relation between $\Omega$ and other involved geometric structures is visualized in Fig.~\ref{fig:geometry}.}

\begin{figure}[h!]
\begin{center}
\includegraphics[scale=0.3]{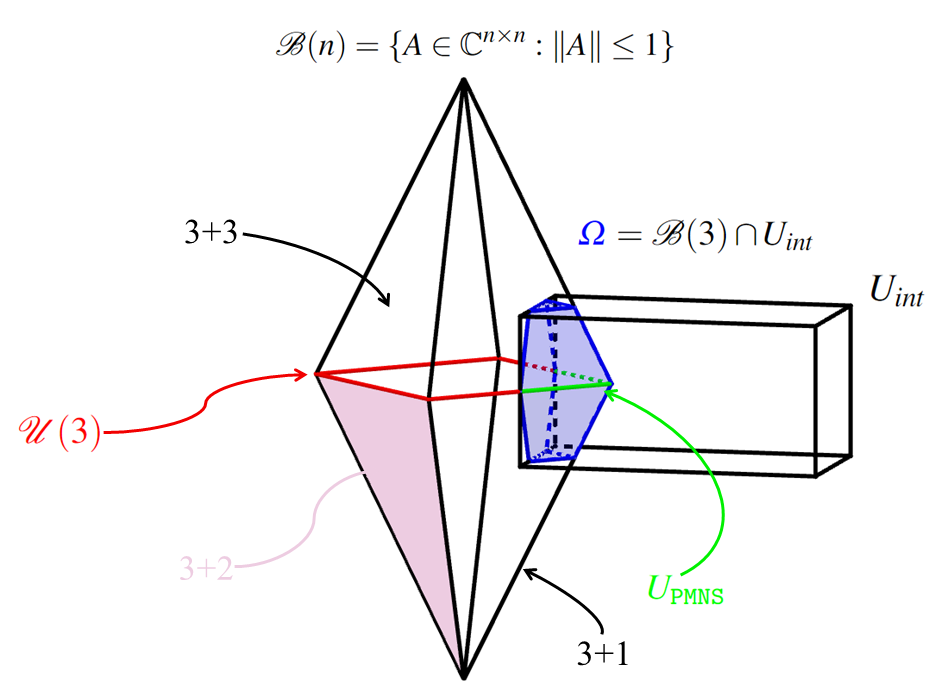} 
\end{center}
\caption{Schematic visualization of the region of physically admissible mixing matrices as an intersection of $\mathcal{B}(3)$ and $U_{int}$. The double pyramid shape corresponds to the unit ball of a spectral norm. Its middle circumference, in red, represents its extreme points, i.e. $\mathcal{U}(3)$ group. Its edges and sides represent contractions with a minimal unitary extension, $4 \times 4$ and $5 \times 5$, respectively. Whereas the interior of $\mathcal{B}(3)$ corresponds to the contraction that minimally can be extended to $6 \times 6$ unitary matrices. The cuboid represents a hypercube of the interval matrix $U_{int}$. At the intersection of these two structures is the $\Omega$ region, in blue, and the set of \texttt{PMNS} mixing matrices is highlighted in green.}
\label{fig:geometry}
\end{figure}

The geometry of $\mathcal{B}(n)$ is strictly connected to the geometry of symmetric gauge functions \cite{stewart1990matrix}.
\begin{definition}
A function $\Phi:\mathbb{R}^{n} \rightarrow \mathbb{R}$ is a symmetric gauge function if it satisfies the following conditions
\begin{enumerate}
\item $\Phi$ is a vector norm,
\item For any permutation matrix $P$ we have $\Phi(Px) = \Phi(x)$,
\item $\Phi(\vert x \vert) = \Phi(x)$.
\end{enumerate} 
\end{definition}

Von Neumann proved that symmetric gauge functions and unitarily invariant norms \eqref{uni_inv_norm} are connected to each other \cite{vonneumann_1937}, namely
\begin{theorem}
$\Vert \cdot \Vert$ is a unitary invariant norm if and only if there exists a symmetric gauge function $\Phi$ such that $\Vert A \Vert = \Phi(S(A))$ for all $A \in \mathbb{C}^{n \times n}$, where $S(A)$ is the set of singular values of $A$.
\end{theorem}
The spectral norm is a unitarily invariant norm and its corresponding symmetric gauge function is an infinite norm, i.e.
\begin{equation}
\Phi_{\infty}(x) = max \lbrace \vert x_{1} \vert ,  \vert x_{2} \vert, \dots, \vert x_{n} \vert \rbrace.
\end{equation}
The unit ball of the infinite norm is a hypercube
\begin{equation}
\mathcal{B}_{\infty}(n) = \lbrace x \in \mathbb{R}^{n}: \Phi_{\infty}(x) \leq 1\rbrace = \left[ -1,1 \right]^{n}.    
\end{equation}
The Von Neuman's relation between unitary invariant norms and symmetric gauge functions is also reflected in the geometry of the corresponding unit balls. The characterization of the extreme points and facial structure of unit balls of unitarily invariant norms by the corresponding structure of unit balls of symmetric gauge functions have been studied in \cite{ZIETAK198857, so_1990, DESA1994349, DESA1994429, DESA1994451}. Faces and extreme points are defined as follows \cite{ROCKAFELLAR, fund_of_ca}
\begin{definition}
Let $C \in \mathbb{R}^{n}$ be a convex set. A convex set $F \subseteq C$ is called a face of $C$ if for every $x \in F$ and every $y,z \in C$ such that $x \in (y,z)$, we have $y,z \in F$. 
\end{definition}
\begin{definition}
The zero dimensional faces of a convex set $C$ are called extreme points of $C$. Thus a point $x \in C$ is an extreme point of $C$ if and only if there is no way to express $x$ as a convex combination $(1-\lambda)y + \lambda z$ such that $y,z \in C$ and $0<\lambda <1$, except by taking $x=y=z$.  
\end{definition}
It appears that the extreme points of the $\mathcal{B}(n)$ are exactly unitary matrices. This result has also been obtained in a more general setting by Stoer \cite{Stoer:1964}. The facial structure of the $\mathcal{B}(n)$ is given in the following theorem \cite{so_1990,DESA1994451} 
\begin{theorem}
$\mathcal{F}$ is a face of $\mathcal{B}(n)$ if and only if there exist $0 \leq r \leq n$ and unitary matrices $U$ and $V$ such that
\begin{equation}
\label{u_ball_faces}
\mathcal{F} = \lbrace
U
\left(
\begin{array}{cc}
I_{r} & 0 \\
0 & A
\end{array}
\right)
V:
A \in \mathcal{B}(n-r) \rbrace.
\end{equation}
\end{theorem}

As the $\Omega$ region in Eq.~(\ref{omega_intersection}) is a subset of $\mathcal{B}(3)$, its geometric structure is inherited from $\mathcal{B}(3)$. Thus, the facial structure of the $\Omega$ region is the same as for $\mathcal{B}(3)$ with restriction of parameters of unitary matrices $U$ and $V$ to experimental results in Eq.~\eqref{pmns_limits} and with established ranges of singular values 

\begin{equation}
\begin{split}
& \text{CP invariant scenario: }\\ 
&\lbrace \sigma_{1} = 0.95954, \sigma_{2} = 0.88186, \sigma_{3} = 0.84189 \rbrace, \\
& \text{General scenario: }\\ 
&\lbrace \sigma_{1} = 0.95592, \sigma_{2} = 0.84112, \sigma_{3} = 0.70275 \rbrace.   
\end{split}
\label{sing_val_phys}
\end{equation}

The faces of $\mathcal{B}(3)$ defined in Eq.~ \eqref{u_ball_faces} do not correspond entirely to physically interesting subsets in Eqs.~\eqref{omega_1}-\eqref{omega_3} of $\Omega$. Namely, higher-dimensional faces contain lower-dimensional faces, e.g. for $r=1$ the face contains not only matrices with two singular values strictly less than one, but also unitary matrices and contractions with only one singular value strictly less than one. In other words faces of $\mathcal{B}(3)$ comprise matrices from different subsets of $\Omega$. To restrict faces to subsets containing only matrices with the specific number of singular values strictly less than one, we can use the notion of the relative interior \cite{ROCKAFELLAR}.
\begin{definition}
The relative interior of a convex set $C \subset \mathbb{R}^{n}$, which is denoted by $ri(C)$, is defined as the interior which results when $C$ is regarded as a subset of its affine hull $\text{aff}(C)$.
\end{definition}
\noindent In this definition by the affine hull of the set $C$ we understand the set of all finite affine combinations of elements of $C$ \cite{leonard:2015}, i.e. $\text{aff}(C) = \lbrace \sum_{i=1}^{k} \alpha_{i} x_{i} : x_{i} \in C, \sum_{i=1}^{k} \alpha_{i} = 1  \rbrace$. 
In that way, the subsets of $\mathcal{B}(n)$ corresponding to different minimal unitary extensions are the relative interiors of $\mathcal{F}$, i.e. subsets of faces for which singular values of the $A$ submatrix are strictly smaller than one.

\begin{definition}
The subsets $\Omega_{1},\dots,\Omega_{4}$ of the $\Omega$ region are relative interiors of the faces $\mathcal{F}$ of $\mathcal{B}(3)$ for $r=2,1,0,3$, respectively, with parameters of unitary matrices $U$ and $V$ restricted by experimental data and with allowed ranges of singular values.
\end{definition}

There is another way to characterize subsets of the $\Omega$ region, namely in terms of Ky-Fan k-norms.
\begin{definition}
For a given matrix $A \in \mathbb{C}^{n\times n}$ the Ky-Fan k-norm is defined as the sum of $k$ largest singular values
\begin{equation}
\Vert A \Vert_{k} = \sum_{i=1}^{k}\sigma_{i}(A), \text{  for k=1,\dots,n}.    
\end{equation}
\end{definition}
In particular for matrices in $\mathbb{C}^{3 \times 3}$ the three possible Ky-Fan norms are
\begin{equation}
\begin{split}
&\Vert A \Vert_{1} = \sigma_{1}(A) \text{  (spectral norm)}, \\
&\Vert A \Vert_{2} = \sigma_{1}(A) + \sigma_{2}(A), \\
&\Vert A \Vert_{3} = \sigma_{1}(A) + \sigma_{2}(A) + \sigma_{3}(A) \text{  (nuclear norm)}.
\end{split}
\end{equation}
Let us define for $k=1,\dots,3$ the following sets
\begin{equation}
\begin{split}
&S_{k}(r)=\lbrace A \in \mathbb{C}^{n\times n}: \Vert A \Vert_{k} = r  \rbrace,  \\
&A_{k}(r_{1},r_{2}) = \lbrace A \in \mathbb{C}^{n\times n}: r_{1} \leq \Vert A \Vert_{k} < r_{2}  \rbrace,
\end{split}    
\end{equation}
i.e. we defined the sphere of radius $r$ and the annulus with radii $r_{1}$ and $r_{2}$ for Ky-Fan norms centered at the origin. Then, the subsets of the $\Omega$ can be defined as
\begin{equation}
\begin{split}
&\Omega_{1} = S_{1}(1) \cap S_{2}(2) \cap A_{3}(2+\sigma_{3min},3), \\
&\Omega_{2} = S_{1}(1) \cap A_{2}(1+\sigma_{2min},2) \cap A_{3}(1+\sigma_{2min}+\sigma_{3min},3), \\
&\Omega_{3} = A_{1}(\sigma_{1min},1) \cap A_{2}(\sigma_{1min}+\sigma_{2min},2) \; \cap  \\ & \;\;\;\;\;\;\;\;\;\;\; A_{3}(\sigma_{1min}+\sigma_{2min}+\sigma_{3min},3), \\
&\Omega_{4} = S_{1}(1) \cap S_{2}(2) \cap S_{3}(3),
\end{split}    
\end{equation}
where $\sigma_{i min}$ for $i=1,2,3$ are lower limits of singular values allowed by current experimental data \eqref{sing_val_phys}.
 
We have described the subsets of the $\Omega$ region corresponding to a different number of additional neutrinos as relative interiors of the unit ball of the spectral norm with restricted range of parameters. In this way we established a correspondence between the geometry of $\Omega$ and $\mathcal{B}(3)$. This will allow us to study the further properties of the region of physically admissible mixing matrices by geometric tools used for unit balls of matrix norms and gauge symmetric functions.

\section{Physically admissible mixing matrices as a convex combination of PMNS matrices}
The $\Omega$ region (\ref{eq:omegaconv}) is defined as the convex hull of unitary \texttt{PMNS} mixing matrices or equivalently as the set of a finite convex combination of \texttt{PMNS} mixing matrices. In this way, we claim that every physically admissible mixing matrix can be represented as a fine convex combination of unitary \texttt{PMNS} matrices. This agrees with the Krein-Milman theorem which states \cite{Krein1940,leonard:2015}

\begin{theorem}
Let $C \subset \mathbb{R}^{n}$ be a nonempty compact convex set, and let $ext(C)$ be the set of extreme points of $C$, then
\begin{equation}
C = \overline{conv}(ext(C)),    
\end{equation}
\end{theorem}
As we have discussed, the extreme points of the unit ball of the spectral norm are unitary matrices and as the $\Omega$ is a subset of the $\mathcal{B}(3)$, the above theorem justifies our definition. However, this theorem does not put any restriction on the number of extreme points necessary to construct any point of a convex set as a convex combination of its extreme points. The upper bound for this number has been given by Carath\'eodory \cite{CarathodoryberDV,leonard:2015}
\begin{theorem}
If $K \subset \mathbb{R}^{n}$, then each point of $conv(K)$ is a convex combination of at most $n+1$ points of $K$. 
\end{theorem}
The natural question arises: What is the Carath\'eodory number for the $\mathcal{B}(n)$ and $\Omega$ which are embedded in $\mathbb{C}^{n^2} \simeq \mathbb{R}^{2n^2}$? In the physically interesting case where $n=3$ according to the Carath\'eodory theorem we would need $19$ unitary matrices. However, we will prove that for $\mathcal{B}(3)$ this number can be significantly reduced.

\begin{prop}
The Carath\'eodory's number for the \\
$conv(\mathcal{U}(3))$=$\mathcal{B}(3)$ is $4$.
\end{prop}
\begin{proof}
Let $B_{\infty}= \lbrace x \in \mathbb{R}^{3} : \Vert x \Vert_{\infty} \leq 1 \rbrace$ be a unit ball of the infinite norm in $\mathbb{R}^{3}$, i.e. the cube $[-1,1]^{3}$. The extreme points of the $B_{\infty}$ are vertices of the cube, i.e. vectors $v_{j} = (\pm 1, \pm 1, \pm 1)^{T}$ for $j=1,\dots,8$. Let $\psi: B_{\infty} \rightarrow \mathbb{M}_{3 \times 3}$ be a mapping which sends a vector from the unit ball $B_{\infty}$ into the diagonal matrix. Then, the $\psi$ sends the extreme points of the $B_{\infty}$ into the diagonal unitary matrices $U_{j}=diag(\pm 1, \pm 1, \pm 1)$ for $j=1,\dots,8$. The Carath\'eodory's number for the cube is 4. Thus, every point in $B_{\infty}$ can be written as the convex combination of at most 4 extreme points $v_{j}$. In particular every point of the positive octant can be written in this way. This means that every diagonal matrix $D \in \mathbb{M}_{3 \times 3}$ with diagonal elements in $[0,1]$ can be written as convex combination of at most 4 diagonal unitary matrices $U_{j}$, i.e. $D= \sum_{i=1}^{4} \alpha_{i} U_{i}, \text{ with } \alpha_{i} \geq 0 \text{ and } \sum_{i=1}^{4}\alpha_{i} = 1$. Now, let $A$ be a contraction with a singular value decomposition $A=W D V^{\dag}$, where $W$ and $V$ are unitary matrices. This gives
\begin{equation}
A=W D V^{\dag} = \sum_{i=1}^{4} \alpha_{i} W U_{i} V^{\dag}.    
\end{equation}
As the $conv(U(3))$=$\mathcal{B}(3)$ is the set of all $3 \times 3$ contractions, this completes the proof. 
\end{proof}

As an immediate consequence of this proposition and the construction used in the proof, {\it matrices from the $\Omega_{2}$ subset, i.e. with two singular values strictly less than one, can be constructed as the convex combination of 3 unitary matrices}. {\it Whereas,  matrices from the $\Omega_{1}$ subset, i.e. with only one singular value strictly less than one, can be constructed as the convex combination of two unitary matrices}.

Following the idea of Stoer \cite{Stoer:1964}, we will show how to construct contractions with two and one singular values strictly less than one as a convex combination of three and two unitary matrices, respectively. Let us take the following diagonal matrix
\begin{equation}
D_{1} =
\left(
\begin{array}{ccc}
1 & 0 & 0 \\
0 & a & 0 \\
0 & 0 & b
\end{array}
\right),
\end{equation}
where $a,b<1$. It can be written as the following sum
\begin{equation}
\begin{split}
&D_{1} = \\
& \frac{1-a}{2}\left(
\begin{array}{ccc}
1 & 0 & 0 \\
0 & -1 & 0 \\
0 & 0 & -1
\end{array}
\right)
+
\frac{a-b}{2}\left(
\begin{array}{ccc}
1 & 0 & 0 \\
0 & 1 & 0 \\
0 & 0 & -1
\end{array}
\right)
+
\frac{1+b}{2}\left(
\begin{array}{ccc}
1 & 0 & 0 \\
0 & 1 & 0 \\
0 & 0 & 1
\end{array}
\right).
\end{split}
\end{equation}
Now let us take another diagonal matrix. This time with only one diagonal element strictly less than one
\begin{equation}
D_{2} =
\left(
\begin{array}{ccc}
1 & 0 & 0 \\
0 & 1 & 0 \\
0 & 0 & a
\end{array}
\right),
\end{equation}
where $a<1$. The $D_{2}$ matrix can be written as
\begin{equation}
D_{2}=
\frac{1-a}{2}\left(
\begin{array}{ccc}
1 & 0 & 0 \\
0 & 1 & 0 \\
0 & 0 & -1
\end{array}
\right)
+
\frac{1+a}{2}\left(
\begin{array}{ccc}
1 & 0 & 0 \\
0 & 1 & 0 \\
0 & 0 & 1
\end{array}
\right).
\end{equation}
Multiplying $D_{1}$ and $D_{2}$ matrices from left- and right-hand side by unitary matrices, we end up with a singular value decomposition of a given matrix with singular values gathered in $D_{1}$ and $D_{2}$, respectively.

{\it As a result contractions with two and one singular values strictly smaller than one can be written as convex combinations of unitary matrices with singular values encoded in coefficients of the combination.  
}

These results will be used for extending phenomenological studies on the light-heavy neutrino mixings undertaken in \cite{Flieger:2019eor} to the 3+2 and 3+3 scenarios.

\section{Volume}
Lie groups  are also manifolds \cite{tu2010introduction}, i.e. they pose geometric structures. Thus, we can associate with them geometrical properties such as the surface area, also called the volume. Two very important groups in physics fall into this category, namely an orthogonal group, and its complex counterpart: a unitary group. These groups are also very important in neutrino physics as the mixing matrix is either orthogonal or, if the CP phase is non-zero, unitary.
In Tab.~\ref{tab:volume} we gathered the list of structures for which we will calculate the volume in this section. The table is split into purely mathematical objects and those restricted by experiments.

\begin{table}
  \begin{center}
  \begin{tabular}{|c|c|}
    \hline
   CP-conserving & CP-violating\\
    \hline
\multicolumn{2}{|c|}{Total volumes} \\
   $\mathcal{SO}(3) \subset \mathcal{O}(3) \subset \tilde{\mathcal{B}}(3)$ &  $\mathcal{SU}(3) \subset \mathcal{U}(3) \subset \mathcal{B}(3)$ \\
\multicolumn{2}{|c|}{Experimentally restricted volumes} \\
$O_{\texttt{PMNS}} \subset \tilde{\Omega} = \tilde{\mathcal{B}}(3) \cap O_{int}$ & $U_{\texttt{PMNS}} \subset \Omega = \mathcal{B}(3) \cap U_{int}$ \\
    \hline    
  \end{tabular} 
  \caption{Total and experimentally restricted volumes for different structures considered in this work and their mutual relations.}
  \label{tab:volume}
  \end{center}
\end{table}

\subsection{CP-conserving case}

The set of all orthogonal matrices of dimension $n \times n$, i.e. $\mathcal{O}(n) = \lbrace O \in \mathbb{R}^{n \times n} : OO^{T} = I\rbrace$, is an example of a Stiefel manifold \cite{muirhead2005aspects}. As the orthogonal matrices have $\frac{n(n-1)}{2}$ independent parameters, the Stiefel manifold of the orthogonal group is a $\frac{n(n-1)}{2}$ dimensional manifold embedded in $n^{2}$ space. We can associate to it a volume which is expressed as the Haar measure over the orthogonal group \cite{muirhead2005aspects,Marinov_1980,Marinov_1981,Boya_2003,zhang2017volumes}
\begin{equation}
vol(\mathcal{O}(n)) = \int_{\mathcal{O}(n)} [O^{T}dO]^{\wedge},
\end{equation}
where $[O^{T}dO]^{\wedge}$ denotes the wedge product of the matrix $O^{T}dO$ and $dO$ is the matrix of the differentials of the orthogonal matrix $O$. This volume can be expressed in the following compact form \cite{ponting,zhang2017volumes}
\begin{equation}
vol(\mathcal{O}(n)) = \frac{2^{n} \pi^{\frac{n^{2}}{2}}}{\Gamma_{n}(\frac{n}{2})} = 
\frac{2^{n} \pi^{\frac{n(n+1)}{4}}}{\prod_{k=1}^{n} \Gamma( \frac{k}{2})}, 
\end{equation}
where $\Gamma_{n}(x) = \pi^{\frac{n(n-1)}{4}} \Gamma(x) \Gamma\left(x-\frac{1}{2}\right) \dots \Gamma\left(x-\frac{n-1}{2}\right)$.
In the case interesting from the neutrino physics perspective, i.e. for $n=3$, this gives
\begin{equation}
vol(\mathcal{O}(3)) = 16 \pi^{2}.
\end{equation}
However, as the determinant of the \texttt{PMNS} mixing matrix is equal to 1, it belongs to even a smaller subset, namely the special orthogonal group $\mathcal{SO}(3)$. The special orthogonal group is a subgroup of $\mathcal{O}(3)$ and its volume is half of the volume of the orthogonal group, i.e.
\begin{equation}
vol(\mathcal{SO}(3)) = 8 \pi^{2}.    
\end{equation}
Moreover, the \texttt{PMNS} matrix does not cover the entire range of parameters and hence we must start from $O^{T}dO$ in order to calculate the volume of this submanifold. Taking the standard \texttt{PMNS} parametrization \eqref{pmns_para} we get
\begin{small}
\begin{equation}
\begin{split}
&O^{T}dO = \\
&\left(
\begin{array}{ccc}
0 & d\theta_{12} + s_{13}d\theta_{23} & c_{12}d\theta_{13} - c_{13}s_{12}d\theta_{23} \\
-d\theta_{12} - s_{13}d\theta_{23} & 0 & c_{13}c_{12}d\theta_{23} + s_{12}d\theta_{13} \\
-c_{12}d\theta_{13} + c_{13}s_{12}d\theta_{23} & -c_{13}c_{12}d\theta_{23} - s_{12}d\theta_{13} & 0
\end{array}   
\right).
\end{split}
\end{equation}
\end{small}
The wedge product of the independent elements of this matrix is equal to
\begin{equation}
[O^{T}dO]^{\wedge} = \cos(\theta_{13})d\theta_{12}d\theta_{13}d\theta_{23}.    
\end{equation}
Thus, the volume of \texttt{PMNS} matrices is given by
\begin{equation}
vol(\texttt{PMNS}) = \int_{\theta_{12_{L}}}^{\theta_{12_{U}}} \int_{\theta_{13_{L}}}^{\theta_{13_{U}}} \int_{\theta_{23_{L}}}^{\theta_{23_{U}}} = \cos(\theta_{13})d\theta_{12}d\theta_{13}d\theta_{23},
\end{equation}
which with the current experimental limits on $\theta_{12}, \theta_{13}$ and $\theta_{23}$ \eqref{pmns_limits} gives
\begin{equation}
vol(\texttt{PMNS}) =2.2667 \times 10^{-4}. 
\label{pmns_volume}
\end{equation}
As we can see the \texttt{PMNS} matrices contribute only in a small portion
to the entire $\mathcal{O}(3)$.

\subsection{CP-violating case}

The unitary group $\mathcal{U}(n)$, i.e the group of all unitary matrices $\mathcal{U}(n) =\lbrace U \in \mathbb{C}^{n \times n} : UU^{T} = I\rbrace$, is another example of Stiefel manifold. Similarly, as for the orthogonal group, the volume of the unitary group is given by the Haar measure over the group
\begin{equation}
vol(\mathcal{U}(n)) = \int_{U(n)} [U^{\dag}dU]^{\wedge},    
\end{equation}
where $[U^{\dag}dU]^{\wedge}$ denotes the exterior product of the matrix $U^{\dag}dU$ and $dU$ is the matrix of the differentials of the unitary matrix $U$. The volume of the unitary group can be expressed in a compact form as \cite{ponting,Marinov_1980,Marinov_1981,Boya_2003,zhang2017volumes}
\begin{equation}
vol(\mathcal{U}(n))=\frac{2^{n} \pi^{n^{2}}}{\tilde{\Gamma}_{n}(n)}=\frac{2^{n} \pi^{\frac{n(n+1)}{2}}}{1!2!\dots(n-1)!},   
\end{equation}
where $\tilde{\Gamma}_{n}(x) = \pi^{\frac{n(n-1)}{2}} \Gamma(x) \Gamma(x-1) \dots \Gamma(x -n +1)$.
Thus, the volume of the $3\times 3$ unitary matrices equals
\begin{equation}
vol(\mathcal{U}(3)) = 4 \pi^{6}.    
\end{equation}
Moreover, the determinant of the \texttt{PMNS} matrix is equal to one, which means it belongs to the special unitary group $\mathcal{SU}(n)$. The volume of the $\mathcal{SU}(n)$ written in the compact form is \cite{Marinov_1980,Marinov_1981,Boya_2003}
\begin{equation}
vol(\mathcal{SU}(n)) = \frac{2^{\frac{(n-1)}{2}}\pi^{\frac{(n-1)(n+2)}{2}}}{1!2!\dots(n-1)!}. 
\end{equation}
For the physically interesting dimension, i.e. $n=3$ this volume is equal to
\begin{equation}
vol(\mathcal{SU}(3)) = \sqrt{3} \pi^{5}.    
\end{equation}

The \texttt{PMNS} mixing matrix with non-zero CP phase, however, has a restricted set of parameters \eqref{dirac_para},\eqref{majorana_para} and   
ranges of these parameters are confined by experiments \eqref{pmns_limits}. Hence, in order to calculate its volume, it is necessary to start from a specific parametrization of the mixing matrix \eqref{pmns_para}. It can be done in the same way as for its real counterpart, i.e. by calculating the wedge product of the matrix $U^{\dag}dU$. However, for the complex matrices, it is much more complicated. Alternatively, it can be calculated by determining the Jacobian matrix of the \texttt{PMNS} matrix in the parametrization \eqref{pmns_para}
\begin{equation}
J = \left( \frac{\partial u_{ij}}{\partial y_{k}} \right), \quad i,j=1,\dots, n  
\end{equation}
and the $y_{k}$ are parameters (for the \texttt{PMNS} matrix $k=1,\dots,4$).
Then, the volume element is multiplied by the  Jacobian $\vert J \vert = \sqrt{det(\frac{1}{2}J^{\dag}J)}$.

The volume of complex \texttt{PMNS} matrices can be calculated in one more way, namely by using the Cartan-Killing metric \cite{kobayashi1996,mkrtchyan2013universal,Zyczkowski_2003}
\begin{equation}
ds^{2}=(V,V)dt^{2},    
\end{equation}
where $(A,B)=\frac{1}{2}Tr(A^{\dag}B)$ is the inner product induced by the Frobenius norm and $V=U^{\dag}dU$. The $V$ is anti-Hermitian and thus $(V,V)=\frac{1}{2}Tr(V^{\dag}V)=-\frac{1}{2}Tr(V^{2})$. 

The Hermitian product of the Jacobian matrix for the \texttt{PMNS} matrix is given by
\begin{equation}
\label{pmns_jacobian}
\frac{1}{2}J^{\dag}J=
\left(
\begin{array}{cccc}
1 & 0 & \sin(\theta_{13})\cos(\delta) & 0 \\
0 & 1 & 0 & 0 \\
\sin(\theta_{13})\cos(\delta) & 0 & 1 & 0 \\
0 & 0 & 0 & \sin^{2}(\theta_{13})
\end{array}
\right).
\end{equation}
Let us look also at the expression for the Cartan-Killing metric
\begin{equation}
ds^{2} = d\theta_{23}^{2} + d\theta_{13}^{2} + d\theta_{12}^{2} + 2\sin(\theta_{13})\cos(\delta) + \sin^{2}(\theta_{13})d\delta^{2}   
\end{equation}
which as expected gives the same matrix as \eqref{pmns_jacobian}. 
Finally, the Jacobian for the \texttt{PMNS} matrices is equal to
\begin{equation}
\vert J \vert = \sqrt{\sin^{2}(\theta_{13}) - \cos^{2}(\delta)\sin^{4}(\theta_{13})}.   
\end{equation}
Thus, the volume of the complex \texttt{PMNS} matrices is given by
\begin{equation}
\begin{split}
&vol(\texttt{PMNS}) = \int_{\texttt{PMNS}} \vert J \vert dV = \\
&=\int \sqrt{\sin^{2}(\theta_{13}) - \cos^{2}(\delta)\sin^{4}(\theta_{13})} d\theta_{23} d\theta_{13} d\theta_{12} d\delta.
\end{split}
\end{equation}
Taking into account current experimental limits for mixing parameters \eqref{pmns_limits} the numerical value for the volume of the complex \texttt{PMNS} mixing matrices is
\begin{equation}
vol(\texttt{PMNS}) = 1.4777 \times 10^{-4}.
\end{equation}
As in the CP conserving case we see that \texttt{PMNS} matrices constitute only a small portion of all unitary matrices.

\textit{This and CP invariant result \eqref{pmns_volume} show already the quality of the neutrino studies.}
However, comparing these results with the volume of quark mixing matrix which is equal to 
\begin{equation}
vol(\texttt{CKM})= 8.81 \times 10^{-14},
\end{equation} we can see that the quark mixing parameters are much more precise.  The \texttt{CKM} mixing matrix can be parametrized in the same way as $\texttt{PMNS}$ in Eq.~(\ref{pmns_para}), the exact values of the \texttt{CKM} parameters are taken from \cite{Zyla:2020zbs}.

\subsection{Scenarios with extra neutrino states}

So far we have established the volume of the neutrino mixing matrices only for the scenario with three known types of neutrinos. However, 
for scenarios with extra neutrino states, it is required to consider the entire $\Omega$ region and not only its extreme points represented by $U_{\texttt{PMNS}}$.
In order to do this, we will use the fact that the region of all physically admissible mixing matrices is a subset of the unit ball in the spectral norm $\mathcal{B}(n) = \lbrace A \in \mathbb{C}^{n \times n}: \Vert A \Vert \leq 1\rbrace$ and for the CP conserving case it is restricted to the real matrices $\tilde{\mathcal{B}}(n) = \lbrace A \in \mathbb{R}^{n \times n}: \Vert A \Vert \leq 1\rbrace$.
Volumes of the $\mathcal{B}(n)$ and $\tilde{\mathcal{B}}(n)$  can be calculated from the singular value decomposition. The differential of the singular value decomposition of a given matrix $A=U\Sigma V^{\dag}$ is equal to
\begin{equation}
dA = dU \Sigma V^{\dag} + U d\Sigma V^{\dag} + U \Sigma dV^{\dag}.     
\end{equation}
By multiplying this from the left-hand side by $U^{\dag}$ and from the right-hand side by $V$ we get
\begin{equation}
dX \equiv U^{\dag}dA V = U^{\dag} dU \Sigma + d\Sigma + \Sigma dV^{\dag} V,    
\end{equation}
which can be rewritten using $dV^{\dag} V = -V^{\dag}dV$ in the following form
\begin{equation}
dX = U^{\dag}dA V = U^{\dag} dU \Sigma + d\Sigma - \Sigma V^{\dag}dV.    
\label{haar_svd}
\end{equation}
The Haar measure is left- and right-invariant, thus $[dA]^{\wedge} = [U^{\dag}dA V]^{\wedge} = [dX]^{\wedge}$. The entrywise analysis of the $dX$ in the CP invariant scenario gives
\begin{equation}
[dX]^{\wedge} = \prod_{i<j} \vert \sigma_{j}^2 - \sigma_{i}^2 \vert \wedge_{i=1}^{n}d\sigma_{i} [O^{T} dO]^{\wedge} [Q^{T} dQ]^{\wedge}.   
\end{equation}
Hence, the volume of the unit ball of the spectral norm in a real case is given by
\begin{equation}
vol(\tilde{\mathcal{B}}(n))=\frac{1}{2^{n}n!} vol(\mathcal{O}(n))^{2} \int_{0}^{1} \prod_{i<j} \vert \sigma_{j}^2 - \sigma_{i}^2 \vert \prod_{k=1}^{n}d\sigma_{k}. 
\label{vol_u_ball_real}
\end{equation}
The inclusion of the factor $\frac{1}{2^{n}}$ assures the uniqueness of the singular value decomposition. 
For the physically interesting dimension $n=3$, we have
\begin{equation}
vol(\tilde{\mathcal{B}}(3)) = \frac{8 \pi^{4}}{45}.    
\label{u_ball_real}
\end{equation}

Similar entrywise analysis of the \eqref{haar_svd} provides the volume element for the singular value decomposition of complex matrices
\begin{equation}
[dX]^{\wedge} = \prod_{i=1}^{n} \sigma_{i} \prod_{i<j} \vert \sigma_{j}^2 - \sigma_{i}^2 \vert^{2} \wedge_{i=1}^{n}d\sigma_{i} [U^{\dag} dU]^{\wedge} [V^{\dag} dV]^{\wedge}   
\end{equation}
Thus, the volume of the unit ball of the spectral norm is given by
\begin{equation}
\begin{split}
vol(\mathcal{B}(n))&=\frac{1}{(2 \pi)^{n}n!} vol(\mathcal{U}(n))^{2} \times \\
&\times \int_{0}^{1} \prod_{k=1}^{n}\sigma_{k} \prod_{i<j} \vert \sigma_{j}^2 - \sigma_{i}^2 \vert^{2} \prod_{k=1}^{n}d\sigma_{k}. 
\label{vol_u_ball_complex}
\end{split}
\end{equation}
The factor $\frac{1}{(2 \pi)^{n}}$ ensures the uniqueness of the singular value decomposition.
In the case interesting from the neutrino physics perspective, i.e. $n=3$, this gives
\begin{equation}
vol(\mathcal{B}(3)) = \frac{\pi^{9}}{8640}.    
\label{u_ball_complex}
\end{equation}
We can use the formulas for the volumes of $\mathcal{B}(3)$ and $\tilde{\mathcal{B}}(3)$ as the basis in the calculation of the volume of the $\Omega$ region. As the $\Omega$ region is defined as the convex hull of the \texttt{PMNS} matrices, to find its volume we need to replace in the formulas \eqref{vol_u_ball_real} and \eqref{vol_u_ball_complex} $vol(\mathcal{U}(n))$ and $vol(\mathcal{O}(n))$ by $vol(\texttt{PMNS})$ in the general and CP conserving case, respectively. Moreover, it is also necessary to restrict ranges of singular values for those allowed by current experimental data \eqref{sing_val_phys}.  
As the result, for the CP invariant scenario, we have the following formula
\begin{equation}
vol(\tilde{\Omega})=\frac{1}{2^{n}n!} vol(\texttt{PMNS})^{2} \int_{\sigma_{min}}^{1} \prod_{i<j} \vert \sigma_{j}^2 - \sigma_{i}^2 \vert \prod_{k=1}^{n}d\sigma_{k}. 
\end{equation}
Taking into account current experimental bounds \eqref{pmns_limits} and allowed ranges for singular values \eqref{sing_val_phys}, the numerical value is equal
\begin{equation}
vol(\tilde{\Omega}) = 6.45 \times 10^{-16}.
\end{equation}
Thus, the $\Omega$ region in the CP invariant case constitutes only $1.84 \times 10^{-20} $ of the unit ball \eqref{vol_u_ball_real}.

For the general case including the CP phase, the formula for the volume of the $\Omega$ region is given by
\begin{equation}
\begin{split}
vol(\Omega)&=\frac{1}{(2 \pi)^{n}n!} vol(\texttt{PMNS})^{2} \times \\
&\times \int_{\sigma_{min}}^{1} \prod_{k=1}^{n}\sigma_{k} \prod_{i<j} \vert \sigma_{j}^2 - \sigma_{i}^2 \vert^{2} \prod_{k=1}^{n}d\sigma_{k}, 
\end{split}
\end{equation}
and its numerical value is
\begin{equation}
vol(\Omega) = 1.12 \times 10^{-18}.
\end{equation}
In the complex case, the contribution of the $\Omega$ region is even smaller than in the CP invariant scenario and it constitutes only $4.34 \times 10^{-27} $ of the respective unit ball in the spectral norm \eqref{vol_u_ball_complex}. It may look like that $vol(\tilde{\Omega})$ is larger than the $vol(\Omega)$, however, we must keep in mind that $\tilde{\Omega}$ and $\Omega$ are structures of different dimensions, and thus cannot be compared directly.  

We have established earlier the characterization of the $\Omega$ region as the intersection of the $\mathcal{B}(3)$ and $U_{int}$ \eqref{omega_intersection}. The $U_{int}$ can be treated as a hyperrectangle in $\mathbb{R}^{9}$ or $\mathbb{C}^{9} \simeq \mathbb{R}^{18}$ respectively for the CP invariant case and the general case. As such, they also are geometric structures with associated volume. This volume is simply the product of the length of its sides, i.e. given intervals. Thus for the CP conserving case, it gives
\begin{equation}
vol(U_{int}) = 2.84 \times 10^{-10}.    
\end{equation}
Whereas when the CP phase is taken into account it is equal to
\begin{equation}
vol(U_{int}) = 2.27 \times 10^{-11}.    
\end{equation}
In \cite{Bielas:2017lok} statistical analysis was performed concerning the amount of physically admissible mixing matrices contained within the interval matrix $U_{int}$. The analysis establishes that for the CP invariant scenario only about $4\%$ of matrices within the interval matrix are contractions. Comparison of volumes gives a similar qualitative result, namely contractions make a small part of the $U_{int}$. The exact calculation reveals that the volume of the $\Omega$ region constitutes $2.3 \times 10^{-4}$ percent of $O_{int}$ in the CP conserving case and $8.12 \times 10^{-6} $ percent of $U_{int}$ for the general complex scenario.

We can check how $vol(\Omega)$ and vol(\texttt{PMNS}) are sensitive to the precision of neutrino parameters. For example, if the range of the CP phase shrinks twice, i.e., we assume  $\delta \in [182.5^{\circ},306^{\circ}]$ then we get
\begin{equation}
\begin{split}
&vol(\texttt{PMNS}) = 7.345 \times 10^{-5}, \\
&vol(\Omega) = 2.759 \times 10^{-19}.
\end{split}
\end{equation}

We can see that $vol(\texttt{PMNS})$ decreased twice, however,  its value  is still a few orders of magnitudes higher than for $vol(\texttt{CKM})$, while $vol(\Omega)$ decreased almost by order of magnitude. 

\section{Summary and outlook}

Neutrino mixings connected with known active three neutrino states can be embedded directly into the $3 \times 3$ unitary matrix. In the case of extra neutrino species, which are still awaiting discovery, the $3 \times 3$ mixing matrix must be extended to a larger unitary matrix.

In general, the physical space of neutrino mixings determined experimentally constitutes a geometric region of finite convex combinations of unitary $3 \times 3$ \texttt{PMNS} mixing matrices. We studied the structure of this region, which is a part of a unit ball of a spectral norm.  

We have described subsets corresponding to a different minimal number of additional neutrinos as relative interiors of faces of this unit ball. This feature of the geometric region allows for the independent phenomenological analysis of 3+n neutrino mixing models. We also gave an alternative characte\-ristic of these subsets in terms of Ky-Fan k-norms.  

 We showed that the Cara\-th\'eo\-dory's number for the $\Omega$ region equals maximally four. 
In 3+1 and 3+2 scenarios, the Cara\-th\'eo\-dory's number is 2 and 3, respectively. This result allows constructing all matrices from the region as the convex combination in an optimal way. 
 We demonstrated a particular construction of contractions with two and one singular values strictly less than one, and with singular values encoded in the coefficients.
Knowing the basis with a minimal number of generating matrices, we will be able to make concrete phenomenological studies of light-heavy neutrino mixings {\it independently} in 3+2 and 3+3 scenarios. It will extend our previous studies using dilation procedure and obtained limits on active-sterile mixings in the 3+1 scenario where one extra neutrino state is present \cite{Flieger:2019eor}. 

We  established the size of the region of the physi\-cally admissible mixing matrices by calculating its volume. 
As zero volume would mean that neutrino parameters are deter\-mined experimentally without errors, its size informs us in some way about the fide\-lity of experi\-mental data extraction. In the case of three neutrino mixing sce\-nario this volume shows that in neutrino physics, compared to the whole space of the mixing para\-meters, a space of possible neutrino mixing para\-meters is already restricted considerably, though when comparing with quark mixings and corresponding volume, the neutrino volume is many orders of magnitude larger.  When additional neutrinos are under consideration the region  narrows down comparing to $\mathcal{B}(3)$ and $U_{int}$  { where $\mathcal{B}(3)$ describes all $3 \times 3$ contractions whereas $U_{int}$ contains experimentally established ranges 
for neutrino mixing matrices and $\Omega$ is the intersection of these two structures.
}
The size of this region will be further squeezed by increasing precision (via increasing  statistics) of future neutrino physics experiments, especially for CP-violating scenarios when the Dirac CP-phase will be determined.

As an outlook, apart from studying unitary extensions of admissible matrices from the $\Omega$ region and the light-heavy neutrino mixings,  
we also plan to apply methods of semidefinite programming and find information about the position of the $\Omega$ region within the unit ball of the spectral norm.
This study will help determine preferable parameter space for future searches for sterile neutrinos in models with different number of extra neutrino states.


\begin{acknowledgements}
We thank Bartosz Dziewit for useful remarks. This work has been supported in part by the Polish National Science Center (NCN) under grant 2020/37/B/ST2/02371 and the Research Excellence Initiative of the University of Silesia in Katowice. 
\end{acknowledgements}

\section{Appendix}

\appendix 

\section{Norms \label{appnorms}}

\begin{definition}
\label{matrix_norm}
A matrix norm is a function $\Vert \cdot \Vert$ from the set of all matrices $\mathbb{M}_{n \times m}$ into $\mathbb{R}$ that satisfies the following properties
\begin{equation}
\begin{split}
&\Vert A \Vert \geq 0 \ and \ \Vert A \Vert=0 \Leftrightarrow A=0, \\
&\Vert \alpha A \Vert = \vert \alpha \vert \Vert A \Vert, \\
&\Vert A+B \Vert \leq \Vert A \Vert + \Vert B \Vert, \\
&\Vert A B \Vert \leq \Vert A \Vert \Vert B \Vert.
\end{split}
\end{equation}
\end{definition} 
In other words, the matrix norm is a vector norm with the additional condition of submultiplicativity. 

There exists an important class of matrix norms consisting of matrix norms which do not change by the unitary multiplication.
\begin{definition}
A matrix norm $\Vert \cdot \Vert$ is called unitarily invariant if for every unitary matrices $U,V$  and a given matrix $A$ it satisfies
\begin{equation}
\Vert U A V \Vert = \Vert A \Vert.
\label{uni_inv_norm}
\end{equation}
\end{definition}

Another important class of matrix norms, called the induced matrix norms, contains matrix norms that are obtained from the vector norms in the following way
\begin{equation}
\Vert A \Vert_{\star} = \max_{\Vert x \Vert_{\star}=1} \Vert Ax \Vert_{\star},   
\end{equation}
where $\Vert \cdot \Vert_{\star}$ stands for the corresponding vector norm. In our case, of particular interest is the matrix norm induced from the Euclidean 2-norm $\Vert x \Vert_{2} = \sqrt{\sum_{i=1}^{n} x_{i}^{2}} = \sqrt{(x,x)}=\sqrt{x^{\dag}x}$ for $x=(x_{1},\dots,x_{n})^{T}$. 
From the Rayleigh quotient $\lambda_{\max}(A) = \max_{\Vert x \Vert_{2} =1} x^{\dag} A x$ \cite{horn_johnson_2012}, we have
\begin{equation}
\begin{split}
\Vert A \Vert_{2}^{2} &= \max_{\Vert x \Vert_{2}=1} \Vert Ax \Vert_{2}^{2} = \max_{\Vert x \Vert_{2}=1} (Ax)^{\dag}Ax = \max_{\Vert x \Vert_{2}=1} x^{\dag}A^{\dag}Ax \\
& = \lambda_{\max}(A^{\dag}A) = \sigma_{1}^{2}(A).
\end{split}
\end{equation}
Thus, the matrix norm $\Vert \cdot \Vert_{2}$ can be defined as the largest singular value of a given matrix. This matrix norm is called an operator norm or spectral norm and will be denoted  as $\Vert \cdot \Vert$. Thus,
\begin{definition}
A spectral norm of a matrix $A \in \mathbb{M}_{n \times m}$ is the matrix norm defined as
\begin{equation}
\Vert A \Vert = \max_{\Vert x \Vert_{2}=1} \Vert Ax \Vert_{2} =\sigma_{1}(A).     
\end{equation}
\end{definition}
\noindent Moreover, the spectral norm is also a unitary invariant norm \eqref{uni_inv_norm}.

\section{Cosine-Sine (CS) decomposition \label{appCS}}

\begin{theorem}
\label{CS_decomposition}
Let the unitary matrix $U \in \mathbb{M}_{(n+m) \times (n+m)}$ be partitioned as
\begin{equation}
U=
\begin{blockarray}{ccc}
n & m & \\ 
\begin{block}{(cc)c}
U_{\texttt{PMNS}} & U_{lh} & \ n \\
U_{hl} & U_{hh} & \ m \\
\end{block}
\end{blockarray},
\end{equation}
If $m \geq n$, then there are unitary matrices $W_{1}, Q_{1} \in M_{n \times n}$ and unitary matrices $W_{2}, Q_{2} \in \mathbb{M}_{m \times m}$ such that
\begin{equation}
\begin{split}
&\left(
\begin{array}{cc}
U_{\texttt{PMNS}} & U_{lh} \\
U_{hl} & U_{hh}
\end{array}
\right)= 
\left(
\begin{array}{cc}
W_{1} & 0 \\
0 & W_{2}
\end{array}
\right)
\left(
\begin{array}{c|cc}
C & -S & 0 \\ \hline
S & C & 0 \\
0 & 0 & I_{m-n}
\end{array}
\right)
\left(
\begin{array}{cc}
Q_{1}^{\dag} & 0 \\
0 & Q_{2}^{\dag}
\end{array}
\right),
\end{split}
\end{equation}
where $C \geq 0$ and $S \geq 0$ are diagonal matrices satisfying $C^{2} + S^{2}=I_{n}$.
\end{theorem}
There exists another form of the CS decomposition which is more important from the neutrino physics perspective. Let $U_{\texttt{PMNS}}$ have the singular value decomposition $U_{\texttt{PMNS}}=W_{1}diag(I_{r},C)Q_{1}^{\dag}$, where $I_{r}$ denotes $r$ singular values equal to one, and $C$ contains singular values that are strictly less than one. The structure of the CS decomposition reveals the intriguing fact, namely the minimal dimension of the unitary dilation of a given contraction is not arbitrary, but is encoded in the number of singular values strictly less than one.
\begin{cor}
\label{min_dim}
The parametrization of the unitary dilation of the smallest size is given by
\begin{equation}
\begin{split}
&
\left(
\begin{array}{cc}
U_{\texttt{PMNS}} & U_{lh} \\
U_{hl} & U_{hh}\
\end{array}
\right)=
\left(
\begin{array}{cc}
W_{1} & 0 \\
0 & W_{2} 
\end{array}
\right)
\left(
\begin{array}{cc|c}
I_{r} & 0 & 0 \\
0 & C & -S \\ \hline
0 & S & C
\end{array}
\right)
\left(
\begin{array}{cc}
Q_{1}^{\dag} & 0 \\
0 & Q_{2}^{\dag}
\end{array}
\right),
\end{split}
\end{equation}
where $r=n-m$ is the number of singular values equal to 1 and $C=diag(\cos \theta_{1},...,\cos \theta_{m})$ with $\vert \cos \theta_{i} \vert <1$ for $i=1,...,m$.
\end{cor}
This is crucial from the physical point of view, since it tells that the minimal number of sterile neutrinos is not arbitrary, but depends on the singular values of the \texttt{PMNS} mixing matrix.

\bibliographystyle{elsarticle-num}       
\bibliography{biblio}   

\end{document}